\documentclass[preprint,showpacs,amssymb,amsmath,latexsym]{revtex4}

\newcommand{\Hi}{\mathcal{H}}

\newcommand{\Tr}{\mathrm{Tr}}

\newcommand{\ket}[1]{| #1 \rangle}
\newcommand{\bra}[1]{\langle #1 |}
\newcommand{\braket}[2]{\langle #1 \vert #2 \rangle}

\newtheorem{Theorem}{Theorem}

\newtheorem{Definition}{Definition}

\newtheorem{Lemma}{Lemma}
\def\QED{\mbox{\rule[0pt]{1.5ex}{1.5ex}}}

\begin{document}

\title{The geometric measure of entanglement for a symmetric pure state with non-negative amplitudes}
\author{Masahito Hayashi $^{1}$, Damian Markham $^{2}$, Mio Murao $^{3,4}$, Masaki Owari $^{5}$, and Shashank Virmani $^{5,6}$}
\affiliation{$^1$ Graduate School of Information Sciences, Tohoku University, Aoba-ku, Sendai 980-8579, Japan \\
$^2$ LTCI CNRS, TELECOM ParisTech, 37/39 rue Dareau, 75014 Paris, France \\
$^3$ Department of Physics, Graduate School of Science, The University of Tokyo, Tokyo 113-0033, Japan \\
$^4$ Institute for Nano Quantum Information Electronics, The University of Tokyo, Tokyo 113-0033, Japan \\
$^5$ Optics Section, Blackett Laboratory and Institute for Mathematical Sciences, Imperial College, London SW7 2AZ, United Kingdom \\
$^6$ SUPA, Department Physics, University of Strathclyde, Glasgow, G4 0NG, United Kingdom \\}

\date{}
\begin{abstract}
In this paper for a class of symmetric multiparty pure states we consider a conjecture related to the geometric measure of entanglement:
``for a symmetric pure state, the closest product state in terms of the fidelity can be chosen as a symmetric product state''.
We show that this conjecture is true for symmetric pure states whose amplitudes are all non-negative in a computational
basis. The more general conjecture is still open.
\end{abstract}

\maketitle

\section{Introduction}
The geometric measure of entanglement, which was first proposed by Shimony \cite{S95}
and extended to multi-partite systems by Wei et al.\cite{WG03},
is one of the most natural entanglement measures for pure states in multi-partite systems,
and has applications in various different topics,
including many body physics \cite{MAVMM08,WDMV05,NDM08}, local discrimination \cite{HMMOV06}, quantum computation \cite{GFE08,MPMNDB},
condensed matter systems\cite{ODV,SOFZ},
entanglement witnesses \cite{GT09,HMMOV08} and the study of quantum channel capacities \cite{Werner02}.

Moreover, the same function (except its normalization) plays important roles in different fields of science apart from physics.
First, the geometric measure of entanglement is nothing but the {\it
injective tensor norm} itself, which appears in theory of operator
algebra \cite{DF92} and is now becoming increasingly important in
theoretical physics - particularly in quantum information \cite{P08}.
Second, in signal processing, especially in the field of multi-way data analysis, high order statistics and independent component analysis (ICA),
this function has been intensively studied under the name of {\it Rank one approximation to high order tensors} \cite{LMV00,ZG01,KR02,WA04,NW07,DL08}.

In spite of its importance, its value has only been determined for limited classes of states with large symmetries,
such as GHZ-states, generalized W-states, and certain families of stabilizer states \cite{WEGM04,MMV07,HMMOV08}.
This is because the geometric measure of entanglement is defined in terms of
the maximum fidelity between the state and a pure product state, and therefore
poses a difficult optimization problem.

In quantum information, there are several entanglement measures which attempt to quantify the ``distance'' between a
quantum state and the set of separable states.
For example, the relative entropy of entanglement and the robustness of entanglement.
For such measures, when a given entangled state is invariant under a certain group action,
we can normally choose a closest separable state as it is invariant under the same group action.
This property of measures helps us to derive an exact value of these measures for such states with group symmetry \cite{VW01}.

On the other hand, even though the geometric measure of entanglement is also defined in term of
distance from a set of product states, it does not possess this property.
In other words, for a given state, closest product states in terms of the fidelity does not inherit its group symmetry.
For example, it is known that there does not exist a translationally invariant closest product state for the translationally-invariant
GHZ-type state $\left ( \ket{01\cdots01}+\ket{10\cdots10} \right ) / \sqrt{2}$.
Nevertheless, many researchers believe that permutation-symmetry is exceptionally inherited by closest product states.
In other words,
there is one prominent conjecture:
{\it for a symmetric pure state, the maximization can be attained by a symmetric product state}.
If this conjecture were to be true, it could vastly reduce the computation
of the geometric measure of entanglement for a symmetric pure state.
To our knowledge this conjecture first appeared in the paper \cite{WG03}, where
it was used in order to propose an analytical formula
for the geometric measure of entanglement for GHZ-states and W-states.
Subsequently this conjecture, and a stronger version
(in which "symmetric" is replaced by "translationally invariant"), was
used in calculations of the geometric measure for states of many body systems \cite{WDMV05,NDM08}.
In the paper \cite{HMMOV08} the authors attempted to prove this conjecture.
However, it remained an open problem.

In this paper, we give a proof of this conjecture for a restricted but large class of symmetric states:
symmetric states whose amplitudes are all non-negative in a given computational basis.
This class involves many famous states like GHZ-states, W-states, Dicke states and also superposition of
these states involving only non-negative coefficients. Our result is hence sufficient to give mathematical
rigour to the computations of the entanglement of types of symmetric pure state
that were presented in \cite{WG03,WEGM04,HMMOV08}.

\section{Definitions and Main result}
Throughout this paper we will treat only finite dimensional Hilbert spaces obtained
from tensor products of a single space $\Hi$.
We start with the definition of the geometric measure of entanglement:
\begin{Definition}
For a state $\ket{\Psi}$, the {\it geometric measure of entanglement} is defined as
\begin{eqnarray} \label{Eqn: Def Eg}
E_g(|\psi\rangle) = \min_{|\Phi\rangle \in {\rm Pro} (\Hi ^{\otimes n})} -\log_2 (|\langle \Phi | \psi \rangle |^2),
\end{eqnarray}
where ${\rm Pro}(\Hi^{\otimes n} )$ is the set of product states on $\Hi^{\otimes n}$.
\end{Definition}
This is the distance between state $|\psi\rangle$ and the
closest product state $|\Phi\rangle$ in terms of fidelity,
and has operational significance in several directions \cite{MAVMM08,WDMV05,NDM08,HMMOV06,GFE08,MPMNDB,GT09,HMMOV08,   Werner02}.
The measure can be extended to the mixed state case
in a natural way via the convex roof method \cite{WG03}.
Several properties of this measure have already been studied and we know
that it has many of the nice properties one might
require from an entanglement measure \cite{S95,WG03,WEGM04,PV}.

The main result of this paper is the following theorem:
\begin{Theorem}
If there exists a basis $ \left \{ \ket{i} \right \}_{i = 1}^{\dim \Hi}$ of $\Hi$
such that a permutation invariant pure state $\ket{\Psi} \in \mathcal{S}_n$ satisfies
$\bra{\Psi}\left(\ket{i_1} \otimes \cdots \otimes \ket{i_n} \right) \ge 0$
for all $i_1, \cdots, i_n$, then a closest product state $\ket{\Phi}$ may be found in the symmetric Hilbert space. More precisely,
\begin{equation}\label{eq: Lemma 1}
E_g(\ket{\Psi})=-\log _2 \max _{\ket{\phi}\in \Hi} |\bra{\phi}^{\otimes n}\ket{\Psi}|^2,
\end{equation}
where $\mathcal{S}_n$ is the symmetric subspace of $\Hi^{\otimes n}$.
In addition we may choose an optimal state $\ket{\phi}$ such that it satisfies $\braket{i}{\phi} \ge 0$ for all $i$.\\
\end{Theorem}
As we have already mentioned, the preconditions of this theorem are satisfied by a large class of symmetric state,
including GHZ and W states. The theorem hence gives a mathematically rigorous proof for the calculations in \cite{WG03,WEGM04,HMMOV08}.
Intriguingly, an identical result to Theorem 1 has also been independently proven by T-C. Wei and S. Severini using methods from the theory of permanents
\cite{Simone}. It will be of interest to identify whether there are hidden similarities to the proofs, or whether
they are truly distinct.

\section{Proof of Main result}

In order to prove Theorem 1, we will need to utilize
Perron-Frobenius Theorem \cite[Theorem 8.3.1]{HJ} and another lemma,
Lemma 1, which we now present and prove.
We use the following notation. First,
for a vector $u \in \mathbb{R}^d$ expressed in a certain privileged basis, we say that $u$ is non-negative if all elements of $u$ are
non-negative in that basis. We denote this using
the notation $u \ge 0$. We use a similar notation for complex vector spaces: a pure state $\ket{v} \in \Hi$ is said to be {\bf non-negative}
if $\braket{i}{v} \ge 0$ for a privileged basis $\{ \ket{i} \}_i$ which we will define shortly.
The non-negativity of a pure state will be denoted by $\ket{v} \ge 0$.\\

For a state $\ket{\Psi} \in \Hi ^{\otimes 2}$,
a product basis $\{\ket{i_1} \otimes \ket{i_2}\}$ gives a natural isomorphism between all states $\ket{\Psi}$ on $\Hi \otimes \Hi$
and all $d \times d$ complex matrices $\Psi$ satisfying $\Tr \Psi ^{\dagger} \Psi=1$, where $d := \dim \Hi$.
If $\{\ket{i_1} \otimes \ket{i_2}\}$ satisfies $\bra{\Psi}(\ket{i_1} \otimes \ket{i_2}) \in \mathbb{R}$,
then the state $\ket{\Psi}$ corresponds to a {\it real matrix} $\Psi$. If such a bipartite state is permutationally
invariant then the matrix will also be symmetric.
The real symmetricity of $\Psi$ implies that 
its largest eigen value $\lambda_1$ is also a singular value so that
\[
\lambda_1 = \sup_{\| w\|=1} w^\dagger \Psi w = 
\sup_{\|w\|=1}\|\Psi w\|.\]
Since $\Psi$ has non-negative elements, the (extended) 
Perron-Frobenius Theorem \cite[Theorem 8.3.1]{HJ} implies that
the largest singular value equals the largest eigen-value 
and the corresponding eigen-vector can be chosen to have non-negative 
elements so that
\begin{equation}\label{condition: Lemma B}
\lambda_1=\sup\{w^T\Psi w|w_k \ge 0,w= 1\}
=\sup\{u^T\Psi v| u_k, v_k \ge 0,\|u\|=\|v\|=1\}.
\end{equation}

Here, we give a lemma:
\begin{Lemma}
If $u$ and $v$ are normalized vectors with non-negative elements and 
$\lambda_1 = u^T\Psi v$, 
then
$w = \frac{u + v}{\|u + v\|}$ is an eigen-vector of $\lambda_1$ 
with non-negative elements.
\end{Lemma}

\begin{proof}
By the Cauchy-Schwarz ineqaulity
\[
\lambda_1 = u^T\Psi v \le \| \Psi v\|\le \lambda_1.
\]
Therefore, equality holds above and by the conditions for equality in 
Cauchy-Schwarz, we must have $cu = \Psi v$. 
Then, the conditions on $u$ and $v$ imply $c = \lambda$.
Applying the same argument to $(\Psi u)^T v$,
we obtain $\Psi u = \lambda v$.
So, 
the relation $\Psi \frac{u + v}{\|u + v\|} = \lambda \frac{u + v}{\|u + v\|}$
 holds.
\end{proof}

Now, we are ready to prove Theorem 1.

\noindent{\it Proof of Theorem 1:~}
We start the proof by noting two important facts:
Firstly, if the statement of this theorem is valid,
then the same statement is valid for a non-normalized state -
the definition of the geometric measure of entanglement $E_g(\ket{\Psi})$ can be
easily extended to a non-normalized state.
Secondly, suppose that the assumption of the theorem is valid.
Then because the amplitudes of $\Psi$ are non-negative we can easily see that
\begin{eqnarray}\label{eq1: proof lemmma 1}
&\quad & \max _{\ket{\Phi} \in {\rm Pro}(\Hi ^{\otimes n})} |\braket{\Phi}{\Psi}| \nonumber \\
&=& \max _{\ket{\Phi}\in {\rm Pro}(\Hi ^{\otimes n})}
\Big\{ \braket{\Phi}{\Psi} \ \Big | \  \bra{i}\otimes \cdots \otimes \bra{i_n} \ket{\Phi} \ge 0, \forall i_1, \cdots, i_n \Big\},
\end{eqnarray}
where ${\rm Pro}(\Hi ^{\otimes n})$ is the set of all product states on $\Hi ^{\otimes n}$.
So, we just need to consider the optimization problem in the right hand side of the above equation.
We will prove Theorem 1 by induction with respect to a number $n$ of tensor copies of the Hilbert space.

For $n=2$, by means of the natural correspondence between bipartite states and matrices,
we derive
\begin{eqnarray}
&\quad & \max _{\ket{u}, \ket{v} \in \Hi} \left \{ \bra{u} \otimes \bra{v} \ket{\Psi} \ \Big | \  \braket{i}{u} \ge 0, \ {\rm and} \ \braket{i}{v} \ge 0, \ \forall i \right \} \nonumber \\
& = & \max _{u, v \in \mathbb{R}^d } \left \{ u^T \Psi v \ \Big | \  u, v \ge 0 \right \}.
\end{eqnarray}
Thus, in the case $n=2$ equation (\ref{eq: Lemma 1}) follows directly from
(\ref{condition: Lemma B}).

Suppose that for all $n \le k$ the statement of this theorem is valid,
and $\ket{\Psi} \in \Hi ^{\otimes k+1}$ satisfies the assumption of the theorem.
Then, since $\ket{\Psi}$ is non-negative, it satisfies Eq.(\ref{eq1: proof lemmma 1}).
Thus, there exists a non-negative product state $\ket{a_1} \otimes \cdots \otimes \ket{a_{k+1}} \ge 0$
satisfying
\begin{equation}
\bra{a_1} \otimes \cdots \otimes \bra{a_{k+1}}\ket{\Psi}=\max _{\ket{\Phi} \in {\rm Pro}(\Hi ^{\otimes k+1})} |\braket{\Phi}{\Psi}|.
\end{equation}
We now define a non-normalized state $\ket{\Psi'_0} \in \Hi ^{\otimes k}$ as $\ket{\Psi'_0} \stackrel{\rm def}{=} I_{\Hi}^{\otimes k}\otimes \braket{a_{k+1}}{\Psi}$ - clearly this state is also non-negative.
Then, this state satisfies
\begin{equation}\label{eq: k partite, lemma 1}
\bra{a_1} \otimes \cdots \otimes \bra{a_{k}}\ket{\Psi'_0}=\max _{\ket{\Phi} \in {\rm Pro}(\Hi ^{\otimes k+1})} |\braket{\Phi}{\Psi}|.
\end{equation}
Now suppose that there exists a non-negative product state $\ket{a'_1} \otimes \cdots \otimes \ket{a'_{k}} \ge 0$
satisfying
$\bra{a'_1} \otimes \cdots \otimes \bra{a'_{k}}\ket{\Psi'_0} > \bra{a_1} \otimes \cdots \otimes \bra{a_{k}}\ket{\Psi'_0}$.
Then, a non-negative product state $\ket{a'_1} \otimes \cdots \otimes \ket{a'_{k}} \otimes \ket{a_{k+1}} \ge 0$ satisfies
$\bra{a'_1} \otimes \cdots \otimes \bra{a'_{k}}\otimes \bra{a_{k+1}} \ket{\Psi} > \max _{\ket{\Phi} \in {\rm Pro}(\Hi ^{\otimes k+1})} |\braket{\Phi}{\Psi}|$.
However, since $\bra{a'_1} \otimes \cdots \otimes \bra{a'_{k}}\otimes \bra{a_{k+1}} \in {\rm Pro}(\Hi ^{\otimes k+1})$,
this would be a contradiction. Hence we obtain
\begin{equation}\label{eq: k partite optimal, lemma 1}
\bra{a_1} \otimes \cdots \otimes \bra{a_{k}}\ket{\Psi'_0}=\max _{\ket{\Phi'} \in {\rm Pro}(\Hi ^{\otimes k})} |\braket{\Phi'}{\Psi'_0}|.
\end{equation}
We now impose the assumption of the induction, that there exists a state $\ket{v_0} \ge 0$ such that
$\bra{v_0}^{\otimes k} \ket{\Psi'_0} = \bra{a_1} \otimes \cdots \otimes \bra{a_{k}}\ket{\Psi'_0} =
\max _{\ket{\Phi'} \in {\rm Pro}(\Hi ^{\otimes k})} |\braket{\Phi'}{\Psi'_0}|$.
From Eq. (\ref{eq: k partite, lemma 1}), we derive
\begin{equation}
\bra{v_0}^{\otimes k} \otimes \bra{a_{k+1}} \ket{\Psi'_0}=\max _{\ket{\Phi} \in {\rm Pro}(\Hi ^{\otimes k+1})} |\braket{\Phi}{\Psi}|
\end{equation}
Now, we define a finite sequence of non-negative states $\{ \ket{C_p^{(0)}}\}_{p=1}^{k+1}$ as
\begin{eqnarray}
 &&  \ket{C_p^{(0)}} := \ket{v_0} \,\,\,\,\,\,\,\,\,\,\,\,\,\,\, {\rm for} \,\,\,\,\, 1 \le p \le k \nonumber \\
 &&  \ket{C_{k+1}^{(0)}} := \ket{a_{k+1}}. \nonumber
\end{eqnarray}
By utilising procedure detailed below, we will use this definition as a starting point for the construction of an infinite sequence of sets of non-negative states
$\left \{ \{ \ket{C_p^{(i)}} \}_{p=1}^{k+1}  \right \}_{i=0}^{\infty}$ satisfying
\begin{equation}\label{eq: optimality of all sequences}
\bra{C_1^{(i)}} \otimes \cdots \otimes \bra{C_{k+1}^{(i)}}\ket{\Psi} =\max _{\ket{\Phi} \in {\rm Pro}(\Hi ^{\otimes k+1})} |\braket{\Phi}{\Psi}|
\end{equation}
for all non-negative integers $i$.
We note that because of the permutation symmetry of $\ket{\Psi}$, there is no significance to the order imposed by $p$.
$\{ \ket{C_p^{(i+1)}}\}_{p=1}^{k+1}$ is defined from $\{ \ket{C_p^{(i)}} \}_{p=1}^{k+1}$ as follows:
We choose a couple of states $\{ \ket{C_{\alpha}^{(i)}}, \ket{C_{\beta}^{(i)}}\}$
from $\{\ket{C_p^{(i)}}\}_{p=1}^{k+1}$ such that
their inner product $\braket{C_{\alpha}^{(i)}}{C_{\beta}^{(i)}}$ is
the least amongst the inner products of all pairs of states selected from $\{ \ket{C_p^{(i)}}\}_{p=1}^{k+1}$.
Then, $\ket{C_{\alpha}^{(i+1)}}$ and $\ket{C_{\beta}^{(i+1)}}$ are defined as
\begin{equation}\label{eq 8 oct}
\ket{C_{\alpha}^{(i+1)}} = \ket{C_{\beta}^{(i+1)}} :=
{\ket{C_{\alpha}^{(i)}} + \ket{C_{\beta}^{(i)}} \over \parallel \ket{C_{\alpha}^{(i)}} + \ket{C_{\beta}^{(i)}}  \parallel}.
\end{equation}
For all other $p \neq \alpha, \beta$, we define $\ket{C_{p}^{(i+1)}}$ as $\ket{C_{p}^{(i+1)}}=\ket{C_{p}^{(i)}}$.
We need to show that the set of non-negative states $\{ \ket{C_p^{(i+1)}} \}_{p=1}^{k+1}$ defined as above actually
satisfies equation (\ref{eq: optimality of all sequences}) for all $i$.
From the permutation symmetry of $\ket{\Psi}$, we can set $\alpha =1$ and $\beta=2$ without losing any generality.
Then, we define a non-negative non-normalized bipartite state $\ket{\Psi'_i} \in \Hi ^{\otimes 2}$ as
$\ket{\Psi'_i}:=I_{\Hi}^{\otimes 2} \otimes \bra{C_3^{(i)}} \otimes \cdots \bra{C_{k+1}^{(i)}}\ket{\Psi}$.
By the same discussion we used to derive equation (\ref{eq: k partite optimal, lemma 1}),
we can conclude that
\begin{equation}
\bra{C_1^{(i)}}\otimes \bra{C_2^{(i)}} \ket{\Psi'_i}= \max _{\Phi \in {\rm Pro}(\Hi ^{\otimes 2})} |\braket{\Phi}{\Psi'_i}|.
\end{equation}
By means of Lemma 1, we obtain
$\bra{C_1^{(i+1)}}\otimes \bra{C_2^{(i+1)}} \ket{\Psi'_i}=\bra{C_1^{(i)}}\otimes \bra{C_2^{(i)}} \ket{\Psi'_i}=
\max _{\Phi \in {\rm Pro}(\Hi ^{\otimes 2})} |\braket{\Phi}{\Psi'_i}|$.
This means that $\{ \ket{C_p^{(i+1)}}\}_{p=1}^{k+1}$ satisfies equation (\ref{eq: optimality of all sequences}) for all $i$.
We are now at a stage where we have a symmetrisation procedure that produces a sequence of product states that
all have the maximal inner product with the entangled state $\Psi$. We must however show that this sequence of
product states converges to a symmetric product state. This is the step that we now address.

%
%

Suppose ${\rm ang} \left ( \ket{u}, \ket{v} \right ) $ is the angle between two single-party states $\ket{u}$ and $\ket{v}$.
We define $\theta _i$ as $\theta _i := \max_{1 \le p, q \le k+1} {\rm arg}( \ket{C_p^{(i)}}, \ket{C_q^{(i)}})$.
Then, by the definition of $\{ \ket{C_p^{(i)}}\}_{p=1}^{k+1}$,
we can easily see $\theta _{i+1} \le \theta _{i}$. Moreover, we can prove $\lim _{i \rightarrow \infty} \theta _{i} =0 $ as follows.

Suppose $\{ \ket{C_p^{(i)}}\}_{p=1}^{k+1}$ satisfies
$\ket{C_p^{(i)}}=\ket{C_q^{(i)}}=\ket{u}$ for all $1 \le p, q \le \xi$, and
$\ket{C_p^{(i)}}=\ket{C_q^{(i)}}=\ket{v}$ for all $\xi+1 \le p, q \le k+1 $.
Without loss of generality, we can assume $\xi \geq (k+1)/2$.
Defining $\eta$ as $\eta := k+1 - \xi$, we can easily see that $\{ \ket{C_p^{(i+\eta)}}\}_{p=1}^{k+1}$ satisfies
$\ket{C_p^{(i+\eta)}}=\ket{C_q^{(i)}}=\ket{u}$ for all $1 \le p, q \le \xi - \eta$, and
$\ket{C_p^{(i+\eta)}}=\ket{C_q^{(i)}}=\ket{u}+ \ket{v} / \| \ket{u}+ \ket{v} \|$ for all $\xi - \eta +1 \le p, q \le k+1 $.
Hence, if $k+1$ is an even number and if $\xi = (k+1)/2$, then $\theta _{i+\eta}=0$.
Otherwise, $\theta _{i+\eta}=\theta _{i}/2$.
Therefore, if $\{ \ket{C_p^{(i)}}\}_{p=1}^{k+1}$ satisfies
$\ket{C_p^{(i)}}=\ket{C_q^{(i)}}=\ket{u}$ for all $1 \le p, q \le \xi$, and
$\ket{C_p^{(i)}}=\ket{C_q^{(i)}}=\ket{v}$ for all $\xi+1 \le p, q \le k+1 $,
then, $\theta _{i+f} \le \theta_{i}/2$,
where $f$ is the largest integer smaller than $k+1/2$.
Since  $\{ \ket{C_p^{(0)}}\}_{p=1}^{k+1}$ actually satisfies the above condition,
we derive $\theta _{h  \cdot f } \le \theta _0 /2^h$ for all positive integers $h$.
Thus, we can conclude $\lim _{i \rightarrow \infty} \theta _{i} =0 $.

A sequence of non-negative states $\{ \ket{C_p^{(i)}}\}_{i=0}^{\infty}$ hence converges to the same non-negative state
$\ket{C^{\infty}}:=\lim _{i \rightarrow \infty } \ket{C_p^{(i)}}$ without depending on $p$.
Since Eq.(\ref{eq: optimality of all sequences}) is valid for all non-negative integers $i$,
by means of the continuity of the inner product,
we obtain $\bra{C^{\infty}}^{\otimes k+1}\ket{\Psi}=\max _{\ket{\Phi} \in {\rm Pro}(\Hi ^{\otimes k+1})} |\braket{\Phi}{\Psi}|$.
That is, the statement is valid for $n=k+1$.
Therefore, by induction with respect of $n$, we have proved the statement of Theorem 1.
\hspace*{\fill}~\QED\par\endtrivlist\unskip

Here, we add one remark concerning the necessity of our Lemma 1.
Indeed from the Perron-Frobenius (PF) Theorem\cite[Theorem 8.3.1]{HJ}, 
we can immediately show the existence of a non-negative $w_0$ 
which satisfies Eq.(\ref{condition: Lemma B}).
However, this is not enough to prove the theorem for the following reasons:
for fixed $\alpha\neq \beta$,
the PF Theorem does enable symmetrization of the closest product state
so that a new closest product state
has the form of
$\ket{C_{\alpha}^{(i+1)}} = \ket{C_{\beta}^{(i+1)}}$.
However, since the symmetrization on a pair of particles $\alpha$ and $\beta$
in principle break the symmetry previously established among other pairs of particles,
simple application of the PF theorem cannot conclude that there
exists a sequence of closest product states converging to a completely symmetric state.
Therefore, in order to demonstrate the convergence
we really need a relation between
$\{\ket{C_{\alpha}^{(i)}}, \ket{C_{\beta}^{(i)}} \}$
and $\{\ket{C_{\alpha}^{(i+1)}}, \ket{C_{\beta}^{(i+1)}} \}$, as in Lemma 1.
Therefore, 
we need Lemma 1 to symmetrize a closest product state and
Lemma 1 is essential for our proof.


\section{Discussion and Conclusions}

In last part of this paper we give several comments on the theorem.
A stronger (and still unproven) version of Theorem could
still be valid, without the assumption of the ``{\it non-negativity}'' of
the symmetric state $\ket{\Psi}$. As we have mentioned in the beginning of the paper,
this stronger conjecture first appeared in Wei et al.'s paper \cite{WG03};
they used this conjecture in order to propose an analytical formula
for the geometric measure of entanglement for GHZ-states and W-states.
While we have not been able to prove the stronger version,
the proof presented here of Theorem 1 applies to W and GHZ states, as they can
be chosen to be non-negative in the sense that we require.

In fact, all specific instances of the geometric measure
calculated in our previous paper \cite{HMMOV08} concern such `non-negative' states,
and so the weaker version of the conjecture, Theorem 1, proved here is sufficient for
those cases (See Section I\hspace{-.1em}I\hspace{-.1em}I.B of  \cite{HMMOV08}).

This weaker version of the conjecture proven above can be useful for calculations of
the geometric measure of entanglement of various multi-partite systems.
For instance, recently several researches have investigated possible connections between the behavior of the geometric measure of entanglement
and existence of quantum phase transition in natural physical systems \cite{WDMV05,NDM08}.
However, it is generally impossible to calculate a value of the geometric measure of entanglement
for such large systems because the definition of the geometric measure involves a large
optimization problem over all product states.
Theorem 1 above provides a way to reduce the size of the optimization problem
in those cases where the state is known to be symmetric and also non-negative.
Actually, in almost all the calculations to date of the geometric measure for ground states,
the possibility of this type of reduction has been assumed.
This paper gives a mathematically rigorous proof of
this type of reduction for a restricted subset of pure states (``{\it the set all non-negative states}'') on a symmetric Hilbert space.

Finally, we mention the possibility of the extension of this theorem for a larger subset of symmetric states.
The logic in the proof of Theorem 1 strongly depends on the reduction of the optimization problem
described by equation (\ref{eq1: proof lemmma 1}). However, a similar reduction is no longer trivial for a state $\ket{\Psi}$
having negative amplitudes.
Moreover, when the state $\ket{\Psi}$ has complex amplitudes, 
Lemma 1 is not valid (although
it is of course clear that applying a local unitary $U \otimes U \otimes ...$ to a `non-negative' state gives
a `non-positive' state for which the conjecture is true).
Hence supplying either a proof or a counterexample to the original stronger statement of Theorem 1 \cite{HMMOV08} is
an interesting open problem.

\section*{Acknowledgments}

We are very grateful to Robert H\"{u}bener and Otfried G\"{u}hne for pointing out the error in
the argument presented in our original paper \cite{HMMOV08}. We also thank Tzu-Shieh Wei
and Simone Severini for discussions.
We are obliged to the editor for pointing out a very simple proof of Lemma 1.
MO was supported by the EPSRC grant
EP/C546237/1, and the European Union Integrated Project Qubit Applications (QAP).
MM acknowledges support from JST and NanoQuine.
MH supported by a Grant-in-Aid for Scientific Research
in the Priority Area `Deepening and Expansion of Statistical Mechanical Informatics (DEX-SMI)', No. 18079014
and a MEXT Grant-in-Aid for Young Scientists (A) No. 20686026. SV acknowledges
support from EU-STREP `Corner' and Quisco.

\end{document}